\documentclass[aps, superscriptaddress,amsmath, nofootinbib]{revtex4}
  \usepackage{amsfonts}
\usepackage{graphicx}
\usepackage[english]{babel}
\usepackage{amsmath}
\usepackage{amssymb}
\usepackage{amsthm}

\newtheorem{proposition}{Proposition}
\newtheorem{lemma}{Lemma}

\newcommand{\dd}[1]{\mathrm{d}#1}

\begin{document}

\title{General solution to Euler-Poisson equations of a free symmetric body by direct summation of power series}

\author{Guilherme Corrêa Silva}
\email{guilherme.jfa@hotmail.com} \affiliation{Depto. de F\'isica, ICE, Universidade Federal de Juiz de Fora, MG,
Brazil}


\date{\today}
\begin{abstract}
Euler-Poisson equations describe the temporal evolution of a rigid body's orientation through the rotation matrix and angular velocity components, governed by first-order differential equations. According to the Cauchy-Kovalevskaya theorem, these equations can be solved by expressing their solutions as power series in the evolution parameter. In this work, we derive the sum of these series for the case of a free symmetric rigid body. By using the integrals of motion and directly summing the terms of these series, we obtain the general solution to the Euler-Poisson equations for a free symmetric body in terms of elementary functions. This method circumvents the need for standard parametrizations like Euler angles, allowing for a direct, closed-form solution. The results are consistent with previous studies, offering a new perspective on solving the Euler-Poisson equations.
\end{abstract}

\maketitle

\section{Introduction}

According to Euler's rotation theorem  \cite{euler1776formulae, arnol2013mathematical, Deriglazov_2023, goldstein2011classical, Landau1976Mechanics, palais2009disorienting}, the temporal evolution of a point ${\bf y}(t)$ of the freely moving rigid body can be presented as follows
\begin{eqnarray}\label{eul.1}
y^i(t) = y_c^i+v_c^i t + R_{ij}(t)x^j(0). 
\end{eqnarray}
In this expression, the term $y_c^i+v_c^i t$ describes the rectilinear motion of the center of mass, $R_{ij}(t)$ is an orthogonal matrix, and $x^j(0)$ are coordinates of the point relative to the center-of-mass at $t=0$. Euler's theorem thereby reduces the problem of describing the motion of a body to searching for the time dependence of the rotation matrix $R_{ij}(t)$. The latter contains all information on the evolution of the body in the Laboratory (fixed in space) frame, in which the body is observed. 

Temporal evolution of the rotation matrix can be obtained from Euler-Poisson equations\footnote{We use the notation 
from work \cite{Deriglazov_2023}. In particular, by $({\bf a}, {\bf b}) = a_ib_i$ and $[{\bf a}, {\bf b}]_i=\epsilon_{ijk}a_j b_k$ we denote the scalar and vector products of the vectors ${\bf a}$ and ${\bf b}$. $\epsilon_{ijk}$ is Levi-Civita symbol in three dimensions, with $\epsilon_{123}=1$. A detailed derivation of the equations (\ref{euler}) and (\ref{poisson}) from the Lagrangian action functional is also given 
in \cite{Deriglazov_2023}.}
\begin{gather}
I\dot{\boldsymbol{\Omega}} = [I\boldsymbol{\Omega}, \boldsymbol{\Omega}],\label{euler}\\
\dot{R}_{ij} = -\epsilon_{jkm}\Omega_kR_{im},\label{poisson}
\end{gather}
where $I$ is the inertia tensor, and the dynamical variables $\Omega_i(t)$ turn out to be components of angular velocity in the body-fixed frame. The formula (\ref{eul.1}) implies \cite{Deriglazov_2023}, that the Euler-Poisson equations should be solved with universal initial data for the rotation matrix: $R_{ij}(0)=\delta_{ij}$. Solutions with other initial conditions do not describe the rigid body motions. The initial data for $\Omega_i(t)$ can be any three numbers: $\Omega_i(0)=\Omega'_i=$\, const. They determine the initial angular velocity of the body. 

In several recent studies, the dynamics of rotating rigid bodies under the influence of both external and internal forces are explored through various physical systems and conditions \cite{amer1, amer2, amer3, amer4, amer5, amer6, amer7}. These systems include bodies subjected to gyrostatic torques, electromagnetic fields, constant and time-varying external forces, resistive effects from viscous media and with a cavity filled with viscous fluid. Analytical solutions for key parameters such as angular velocities, Euler angles, and stability criteria are derived, with phase diagrams and numerical simulations used to assess the motion and stability of these bodies. In addition, the impact of applied forces on the equilibrium and periodicity of the body's motion is addressed, with particular relevance for practical applications in mechanical systems, spacecraft, and satellite technology. These investigations provide valuable insights into the behavior of rigid bodies in rotational motion, contributing to advancements in fields like aerospace, mechanical engineering, and astrophysics.

In this work, we consider a symmetrical body with the principal inertia moments $I_1=I_2\neq I_3$. It is known that in this case, the most general movement of the body is a regular precession: a symmetrical body rotates uniformly around the third axis of inertia while this axis precesses with uniform angular velocity around the axis of conserved angular momentum. There are various possibilities to arrive at this result. 
The traditional way is to solve the equations (\ref{euler}) and (\ref{poisson}) by rewriting them through the Euler angles, and in the Laboratory frame with a third axis directed along the vector of conserved angular momentum \cite{Landau1976Mechanics}. 
Some specific features of rigid body dynamics, that must be taken into account within this method, are discussed in recent works \cite{deriglazov2024asymmetricalbodyexampleanalytical, deriglazov2023problem}.

Another possibility was presented in \cite{deriglazov2023general}, where the explicit form of the rotation matrix through elementary functions was obtained by resolution of equations (\ref{euler}) and (\ref{poisson}) without assuming any kind of parametrization like Euler angles. This was achieved by reducing the original problem to the problem of the motion of a one-dimensional harmonic oscillator under the action of a constant external force. This method also allows one to find particular solutions in elementary functions in several more complex problems, including the cases of dancing top \cite{Deriglazov_2023EP}, Lagrange top \cite{deriglazov2024asymmetricalbodyexampleanalytical} and free symmetric body in stationary and homogeneous electric and magnetic fields \cite{Deriglazov_2024}. 

In the present work, we explore one more possibility based on a remarkable formula, which, in our opinion, is unfairly forgotten and ignored in studies on rigid body dynamics. Euler-Poisson equations belong to the following class of autonomous differential equations 
\begin{equation}\label{sistint}
    \dot{z}^i = h^i(z^j), \qquad i, j= 1, 2,  \ldots p, 
\end{equation}
for determining integral lines $z^i(t)$ of given  vector field $h^i(z^j)$.  It is known \cite{deriglazov2016classical, gantumur2011math, kepley2021constructive, thelwell2012cauchy, folland2020introduction, evans2022partial, deriglazov2024dynamics} that the following family of functions: 
\begin{equation}\label{oprint}
 z^i(t, z^j_0) = e^{th^k(z_0^j)\frac {\partial}{\partial z^k_0}}z^i_0,\ \ \ \ \ \ \ \text{where}\ \ \ \ \ \ e^{th^k(z_0^j)\frac {\partial}{\partial z^k_0}} = \sum_{n=0}^\infty \frac{t^n}{n!}\bigg(h^k(z_0^j)\frac {\partial}{\partial z^k_0}\bigg)^n,
\end{equation}
parameterized by $n$ parameters $z^i_0$, represents their general solution.  This is an immediate consequence of the following  properties of the differential operator $e^{th^k(z_0^j)\frac {\partial}{\partial z^k_0}}$:
\begin{equation}\label{prop1}
    e^{th^k\frac{\partial}{\partial z_0^k}}f(z_0^i) = f(e^{th^k\frac{\partial}{\partial z_0^k}}z_0^i) = f(z^i (t, z_0^j)),
\end{equation}
\begin{equation}\label{prop2}
    \dot{z}^i(t, z_0^j) = \frac{\dd}{\dd t}\big(e^{th^k\frac {\partial}{\partial z^k_0}}z^i_0\big) = e^{th^k\frac{\partial}{\partial z_0^k}}\bigg[\bigg(h^k(z_0^j)\frac{\partial}{\partial z_0^k}\bigg)z_0^i\bigg] = e^{th^k\frac{\partial}{\partial z_0^k}}h^i(z_0^j),
\end{equation}
where $f(z_0^i)$ is an analytic function. Besides, the Cauchy-Kovalevskaya theorem \cite{folland2020introduction, gantumur2011math, kepley2021constructive, thelwell2012cauchy, evans2022partial} guarantees the convergence of the series (\ref{oprint}) in some vicinity of $t=0$. 

So, it is not necessary to directly solve the system (\ref{euler}) and (\ref{poisson}) in one or another way. Instead, we can calculate the explicit form of the series terms (\ref{oprint}) and then try to sum them. As we will show below, for the case of a free symmetric body this turns out to be possible, leading to its rotation matrix in terms of elementary functions. 

The work is organized as follows. In Sect. 2, after presenting our notation and basic equations, we sum up the power series for $\Omega_i(t)$, thereby obtaining the general solution to Euler equations (\ref{euler}) in elementary functions. The summation of the series for $R_{ij}(t)$ of Poisson equations (\ref{poisson}) turns out to be a much more involved task. In Sect. 3 we solve this task in a somewhat trick way, by making use of integrals of motion to represent the components $R_{i1}$ and $R_{i2}$ through the previously found $R_{i3}$.  The direct summation of all components in an independent manner will be done in Sect. 4. Both methods lead to the same final expression for the rotation matrix, coinciding with that obtained in [9]. In the Appendix, for the convenience of the reader, we proved the formulas (\ref{oprint})-(\ref{prop2}).

\section{Notation and basic equations of the problem}\label{1}

Before discussing those solutions to the Euler-Poisson equations that describe the motions of a rigid body, we use the formula (\ref{oprint}) to obtain their general solution with arbitrary initial data. 

Using columns of the matrix $R^T$: $\textbf{G}_i(t) =(R_{i1}, R_{i2}, R_{i3})^T$, the Poisson equations (\ref{poisson}) can be written in the vector-like form: $\dot{\textbf{G}}_i = [\textbf{G}_i, \boldsymbol{\Omega}]$. Let us consider the following Cauchy problem: 
\begin{gather}
    \dot{\boldsymbol{\Omega}} = I^{-1}[I\boldsymbol{\Omega}, \boldsymbol{\Omega}]\label{euler2},\ \ \ \ \ \ \ \ \ \Omega_i(0) = \Omega_i',\\
    \dot{\textbf{G}}_i = [\textbf{G}_i, \boldsymbol{\Omega}],\ \ \ \ \ \ \ \ \ R_{ij}(0) = R_{ij}',\label{poisson2}
\end{gather}
where the initial data  $\Omega_i'$ and $R_{ij}'$ are arbitrary numbers, and $I$ is a diagonal matrix of the following form: $I=diagonal (I_2, I_2, I_3)$.  Solutions to this problem with the data $R_{ij}'=\delta_{ij}$ will describe all possible motions of a symmetrical body that at the initial instant $t=0$ had its inertia axes directed along the laboratory (fixed in space) axes: ${\bf R}_i(0)={\bf e}_i$, see \cite{Deriglazov_2023} for details.

Euler-Poisson equations admit several integrals of motion. They are 
\begin{gather} 
\Omega_3(t)=\Omega'_3=const,\label{ad1}\\
\frac 1 2 \sum_{i = 1}^3 I_i\Omega_i^2=E = const,\label{energia}\\
(RI\boldsymbol{\Omega})_i = \sum_{j=1}^3 I_jR_{ij}\Omega_j =m_i= const.\label{momento}
\end{gather}
For the case of rigid-body motions, $E$ and $m_i$ represent the rotational energy and components of angular momentum, respectively. Besides, taking the equality (\ref{momento}) at $t=0$ and using $R_{ij}(0) = \delta_{ij}$, we get the relation between conserved angular momentum and initial values of angular velocity: $\Omega_i(0) = m_i/I_i$. Using this in (\ref{energia}), we conclude that energy along any trajectory is fixed by the angular momentum:  $E = \frac 1 2 \sum_i m_i^2/I_i$.

Equations (\ref{euler2}) and (\ref{poisson2}) form an autonomous  system of $3+9$ nonlinear first-order differential equations. Their right sides are polynomials and hence represent analytical functions. According to the formula (\ref{oprint}), the unique solution to our initial value problem is given by the series
\begin{gather}
{\boldsymbol{\Omega}}(t, \Omega_j') = \exp{\bigg[{t\bigg({[\textbf{G}_a', \boldsymbol{\Omega}']_b\frac{\partial}{\partial {R'}_{ab}}+ (I^{-1}[I\boldsymbol{\Omega}', \boldsymbol{\Omega}'])_c \frac{\partial}{\partial \Omega_c'}}\bigg)}\bigg]} \boldsymbol{\Omega}'=\exp{\bigg[{t{(I^{-1}[I\boldsymbol{\Omega}', \boldsymbol{\Omega}'])_c \frac{\partial}{\partial \Omega_c'}}}\bigg]} \boldsymbol{\Omega}',\label{omega}\\
{\textbf{G}}_i(t, \Omega_j', R_{kl}') = \exp{\bigg[{t\bigg({[\textbf{G}_a', \boldsymbol{\Omega}']_b\frac{\partial}{\partial {R'}_{ab}}+ (I^{-1}[I\boldsymbol{\Omega}', \boldsymbol{\Omega}'])_c \frac{\partial}{\partial \Omega_c'}}\bigg)}\bigg]}\textbf{G}_i'.\label{alexei}    
\end{gather}
Our goal is to sum these series and try to write them in elementary functions. We start with the analysis of Euler equations. \par

\noindent\textbf{General solution to the Euler equations.} The first term of the power series (\ref{omega}) is just $\boldsymbol{\Omega}'$.
The next term we present as follows:
\begin{equation}\label{1omega}
    \bigg[(I^{-1}[I\boldsymbol{\Omega}', \boldsymbol{\Omega}'])_c\frac{\partial}{\partial \Omega_c'}\bigg]\boldsymbol{\Omega}' = \phi'\begin{pmatrix}
        \Omega_2'\\
        -\Omega_1'\\
        0
    \end{pmatrix}
    = \phi' T_3\boldsymbol{\Omega}',\ \ \ \ \text{where}\ \ \ \  \phi' = (I_2-I_3)\Omega_3'/I_2,\ \ \ \ \text{and}\ \ \ \
    T_3 =
    \begin{pmatrix}
    0 & 1 & 0 \\
   -1 & 0 & 0 \\
    0 & 0 & 0 
\end{pmatrix}
    .
\end{equation}
Above we have three equations and the most interesting of them is the one in the third component. This means that the component $\Omega_3'$ belongs to the kernel of the linear differential operator $(I^{-1}[I\boldsymbol{\Omega}', \boldsymbol{\Omega}'])_c\frac{\partial}{\partial \Omega_c'}$. As a consequence of using the formula (\ref{oprint}), functions of elements of the kernel of the vector field $h^k$ also belongs to its kernel, proposition 3 in the appendix. That is, since functions $f(\Omega_3')$ belongs to the kernel, consequently we have $\big[(I^{-1}[I\boldsymbol{\Omega}', \boldsymbol{\Omega}'])_c\frac{\partial}{\partial \Omega_c'}\big]\phi' = 0$. Then, the next terms become immediate: for $n=2$ we have $\big[(I^{-1}[I\boldsymbol{\Omega}', \boldsymbol{\Omega}'])_c\frac{\partial}{\partial \Omega_c'}\big]^2\boldsymbol{\Omega}' = \big[(I^{-1}[I\boldsymbol{\Omega}', \boldsymbol{\Omega}'])_c\frac{\partial}{\partial \Omega_c'}\big]\big(\phi'T_3\boldsymbol{\Omega}'\big) = \big(\phi'T_3\big)^2\boldsymbol{\Omega}'$, and so on. Performing mathematical induction, the series (\ref{omega}) is rewritten as\footnote{Previously it was mentioned: if the system were linear, the solution would be immediate. That is the solution. More specifically, with the condition $I_1=I_2$, the Euler equations (\ref{euler}) turn out to be an autonomous and linear system with constant coefficients.  This class of systems has an immediate general solution. The interesting thing is that we could arrive at this result from the formula (\ref{oprint}).}

\begin{equation}\label{eulersol}
    {\boldsymbol{\Omega}}(t, \Omega_j') = e^{t\phi'T_3}\boldsymbol{\Omega}' = \bigg(\sum_{n=0}^\infty \frac{t^n}{n!}{\phi'}^nT_3^n\bigg)\boldsymbol{\Omega}'.
\end{equation}
Moreover, the even and odd powers of the matrix $T_3$ satisfies the relations $T_3^{2n} = (-1)^n diag(1, 1, 0)$ (except by $n=0$) and $T_3^{2n+1} = (-1)^{n}T_3$, respectively. Then, considering this in the above sum, we obtain

\begin{equation}\label{domemgan1}
    {\boldsymbol{\Omega}}(t, \Omega_j')=
    \begin{pmatrix}
    \cos{\phi't} & \sin{\phi't} & 0 \\
    -\sin{\phi't} & \cos{\phi't} & 0 \\
    0 & 0 & 1 
    \end{pmatrix}
    \boldsymbol{\Omega}'=
    \begin{pmatrix}
    \Omega_1'\cos\phi' t + \Omega_2'\sin\phi' t \\
    -\Omega_1'\sin\phi' t +\Omega_2'\cos\phi' t \\
     \Omega_3' 
    \end{pmatrix},
\end{equation}
the known general solution to the Euler equations for the case of $I_1=I_2$ in elementary functions. That also represents the general solution to the Euler equations in the case of a free symmetric rigid body, since the unique condition that separates the two situations (of all solutions to the system and the solutions that describes a rigid body) is the data $R_{ij}(0) = \delta_{ij}$. Then, the angular velocity $\boldsymbol{\Omega}(t, \Omega_j')$ rotates around the third axis of the body-fixed frame clockwise with frequency $\phi'$.

\section{General solution to Poisson equations by making use of integrals of motion}\label{1.1}

The energy (\ref{energia}) and angular momentum (\ref{momento}) are integrals of motion of the system (\ref{euler2}), (\ref{poisson2}), even though now we are not considering a physical system or even the condition $R_{ij}(0) = \delta_{ij}$. These integrals are found by physical laws however, they continue being constant relative to $t$ considering the equations (\ref{euler2}), (\ref{poisson2}), without thinking about any physical interpretation. These equations at $t=0$ are read as

\begin{equation}\label{im}
    2E = I_2({\Omega_1'}^2 + {\Omega_2'}^2) + I_3{\Omega_3'}^2, \ \ \ \ \ \ \ \ \ m_i = I_2(\Omega_1'R_{i1}' + \Omega_2'R_{i2}') + I_3\Omega_3'R_{i3}' =  (I\boldsymbol{\Omega}', \textbf{G}_i').
\end{equation}
By direct computation, we obtain\footnote{In sections \ref{1.1} and \ref{2.1} we fixed the notation $\nabla \equiv [{\textbf{G}_a', \boldsymbol{\Omega}']_b\frac{\partial}{\partial {R'}_{ab}}+ (I^{-1}[I\boldsymbol{\Omega}', \boldsymbol{\Omega}'])_c \frac{\partial}{\partial \Omega_c'}}$ to simplify our equations.} $\nabla{E} = 0$ and $\nabla{m}_i = 0$, which endorses a relation between integrals of motion and the kernel of $\nabla$. This relationship also is exploited by proposition 3 in the appendix. Besides, recall that functions of the elements in the kernel as $f(E, m_i)$ or $g(\Omega_3')$ satisfies $\nabla f(E, m_i) = 0 = \nabla g(\Omega_3')$. These facts about integrals of motion will be important to explain the following lemma:

\begin{lemma}\label{l1} Given a numerical matrix $A = diag(A_2, A_2, A_3)$, then the linear differential operator $\nabla$ obeys the relations:

\begin{gather}
    \nabla^{2n}(A\boldsymbol{\Omega}', \textbf{G}_i') = \bigg(A_3 - \frac{I_3}{I_2}A_2\bigg)\bigg[-(-1)^{n}{k'}^{2n}\frac{(M\boldsymbol{\Omega}', \textbf{G}_i')}{{k'}^2}\bigg], \ \ \ \  \text{for}\ \  n\geq 1,\label{even}\\
    \nabla^{2n+1}(A\boldsymbol{\Omega}', \textbf{G}_i') = \bigg(A_3 - \frac{I_3}{I_2}A_2\bigg)\bigg[(-1)^n{k'}^{2n+1}\frac{\Omega_3'[\textbf{G}_i', \boldsymbol{\Omega}']_3}{k'}\bigg], \ \ \ \  \text{for}\ \  n\geq 0,\label{odd}
\end{gather}
where:

\begin{equation}\label{m}
    M = diag\bigg(\frac{I_3}{I_2}{\Omega_3'}^2, \frac{I_3}{I_2}{\Omega_3'}^2, -({\Omega_1'}^2 + {\Omega_2'}^2)\bigg),\ \ \ \ \text{and}\ \ \ \ k'\equiv \sqrt{{\Omega'_1}^{2} + {\Omega'_2}^{2} + \frac{{I_3}^2}{{I_2}^2}{\Omega'_3}^2}.
\end{equation}

\end{lemma}

The above lemma will be used to obtain the coefficients of the series (\ref{alexei}). It is a consequence of the integrals of motion of the EP equations. Firstly, consider the useful facts:\\
\textbf{1}. The function $f(E, \Omega_3') = (2E - I_3{\Omega_3'}^2)/I_2 = {\Omega_1'}^2 + {\Omega_2'}^2$ also belongs to the kernel of $\nabla$. With this, we see that the components of the matrix $M$ satisfies $\nabla M_i =0$. Evidently, since $k' = \sqrt{M_3 - I_3M_2/I_2}$, we also have $\nabla k' = 0$.\\
\textbf{2}. The application of $\nabla$ in the third component in $\textbf{G}_i'$ provide us with $\nabla R_{i3}' = [\textbf{G}_i', \boldsymbol{\Omega}']_3$. Substituting this in $\nabla m_i = I_2\nabla(\Omega_1'R_{i1}' + \Omega_2'R_{i2}') + I_3\Omega_3'\nabla R_{i3}'= 0$, we get the relation $\nabla(\Omega_1'R_{i1}' + \Omega_2'R_{i2}') = - I_3\Omega_3'\nabla R_{i3}'/I_2$. Then, we are able to write a linear combination of $\nabla(\Omega_1'R_{i1}' + \Omega_2'R_{i2}')$ and $\nabla (\Omega_3' R_{i3}')$ as follows

\begin{align}\label{d1}
    \nabla(A\boldsymbol{\Omega}', \textbf{G}_i') = A_2\nabla(\Omega_1'R_{i1}' + \Omega_2'R_{i2}') + A_3\nabla (\Omega_3'R_{i3}') = \bigg(A_3-\frac{I_3}{I_2}A_2\bigg)\Omega_3'[\textbf{G}_i', \boldsymbol{\Omega}']_3.
\end{align}
Next, applying $\nabla$ in (\ref{d1}), by direct computation, we get $\nabla^2(A\boldsymbol{\Omega}', \textbf{G}_i') = (A_3-I_3A_2/I_2)(M\boldsymbol{\Omega}', \textbf{G}_i')$. Since the components in $M$ belong to the kernel of $\nabla$, the matrix $M$ is treated in the same way as $A$ in relative to that operator. Then, we can change $A$ into $M$ in the latter equation and write $\nabla^2(M\boldsymbol{\Omega}', \textbf{G}_i') = (M_3-I_3M_2/I_2)(M\boldsymbol{\Omega}', \textbf{G}_i') = -{k'}^2(M\boldsymbol{\Omega}', \textbf{G}_i')$. This relation can be understood as a geometric progression. Indeed, considering the sequence $a_n = \nabla^{2n}(M\boldsymbol{\Omega}', \textbf{G}_i')$, we have $a_{n+1} = \nabla^{2n}[\nabla^2(M\boldsymbol{\Omega}', \textbf{G}_i')] = -{k'}^2\nabla^{2n}(M\boldsymbol{\Omega}', \textbf{G}_i') = -{k'}^2a_n$. Furthermore, its general term has a well-known form given by $\nabla^{2n}(M\boldsymbol{\Omega}', \textbf{G}_i') = (-{k'})^{2n}(M\boldsymbol{\Omega}', \textbf{G}_i')$. Then, substituting this result in $\nabla^{2n}(A\boldsymbol{\Omega}', \textbf{G}_i') = (A_3-I_3A_2/I_2)\nabla^{2n-2}(M\boldsymbol{\Omega}', \textbf{G}_i')$, we get the expression (\ref{even}) and, after this, applying $\nabla$ in both sides we get (\ref{odd}). In contrast, this lemma also could be obtained by direct computation of terms $\nabla^n(A\boldsymbol{\Omega}', \textbf{G}_i')$ and then made a mathematical induction. That means that, despite this result being a consequence of the integrals of motion, we see that the direct computation of the series (\ref{alexei}) already knows this information.  

Now, we will use this result sum the series (\ref{alexei}). The simplest coefficients to obtain are those of the components $R_{i3}$. Just put $A_2 = 0$ and $A_3 = 1/\Omega_3'$ in the lemma and the result will be the equations below:

\begin{gather}
    \nabla^{2n}{R'}_{i3} = -\frac{1}{\Omega_3'}(-1)^{n}{k'}^{2n}\frac{(M\boldsymbol{\Omega}', \textbf{G}_i')}{{k'}^2},\ \ \ \  \text{for}\ \  n\geq 1,\label{eveni3}\\
    \nabla^{2n+1}{R'}_{i3} = (-1)^n{k'}^{2n+1}\frac{[\textbf{G}_i', \boldsymbol{\Omega}']_3}{k'},\ \ \ \  \text{for}\ \  n\geq 0.\label{oddi3}
\end{gather}
Substituting them in (\ref{alexei}), we get the general solution to the third component of the Poisson equations in elementary functions: 

\begin{align}\label{gen3}
    \nonumber {R}_{i3}(t, \Omega_\gamma', R_{\alpha\beta}') &= {R'}_{i3} -\frac{1}{\Omega_3'} \sum_{n=1}^\infty\frac{t^{2n}}{2n!}(-1)^{n}{k'}^{2n}\frac{(M\boldsymbol{\Omega}', \textbf{G}_i')}{{k'}^2} + \sum_{n=0}^\infty\frac{t^{2n+1}}{(2n+1)!}(-1)^n{k'}^{2n+1}\frac{[\textbf{G}_i', \boldsymbol{\Omega}']_3}{k'}\\
    &={R'}_{i3} + (1-\cos{k't})\frac{(M\boldsymbol{\Omega}', \textbf{G}_i')}{\Omega_3'{k'}^2} + \sin{k't}\frac{[\textbf{G}_i', \boldsymbol{\Omega}']_3}{k'}.
\end{align}
The functions ${R}_{i1}(t, \Omega_\gamma', R_{\alpha\beta}')$ and  ${R}_{i2}(t, \Omega_\gamma', R_{\alpha\beta}')$ we present 
through ${R}_{i3}(t, \Omega_\gamma', R_{\alpha\beta}')$ as follows. Using (\ref{oddi3}) with  $n=0$ and the component $m_i$ of the conserved angular moment (\ref{im}) we get the linear system
\begin{equation}\label{linear}
     \Omega_2'{R'}_{i1} - \Omega_1'{R'}_{i2} = \nabla{R'}_{i3}, \ \ \ \ \ \ \ \ \ I_2(\Omega_1'{R'}_{i1} + \Omega_2'{R'}_{i2}) = m_i - I_3\Omega_3'{R'}_{i3}, 
\end{equation}
for determining the initial data ${R'}_{i1}$ and ${R'}_{i2}$ through ${R'}_{i3}$. Resolving, we get
\begin{equation}\label{c}
    {R'}_{i1} = \frac{1}{I_2({\Omega_1'}^2+ {\Omega_2'}^2)}[I_2\Omega_2'\nabla {R'}_{i3} + (m_i-I_3\Omega_3'{R'}_{i3})\Omega_1'],\ \ {R'}_{i2} = \frac{1}{I_2({\Omega_1'}^2+ {\Omega_2'}^2)}[-I_2\Omega_1'\nabla {R'}_{i3} + (m_i-I_3\Omega_3'{R'}_{i3})\Omega_2'].
\end{equation}
These relations hold for the initial data ${\Omega_1'}^2+ {\Omega_2'}^2\ne 0$. The case $\Omega_1'= \Omega_2'=0$ should be considered separately, see below. 

Let us apply the operator $e^{t\nabla}$ to both sides of these relations among the initial data. With the use of (\ref{prop1}) and (\ref{prop2}), we turn out them into the relations  among the solutions 
\begin{gather}
    {R}_{i1}(t, \Omega_\gamma', R_{\alpha\beta}') = \frac{1}{I_2({\Omega_1'}^2+ {\Omega_2'}^2)}[I_2 {\Omega}_2(t, \Omega_\gamma')\dot{{R}}_{i3}(t, \Omega_\gamma', R_{\alpha\beta}') + (m_i-I_3\Omega_3'{R}_{i3}(t, \Omega_\gamma', R_{\alpha\beta}'){\Omega}_1(t, \Omega_\gamma')],\\
     R_{i2}(t, \Omega_\gamma', R_{\alpha\beta}') = \frac{1}{I_2({\Omega_1'}^2+ {\Omega_2'}^2)}[-I_2\Omega_1(t, \Omega_\gamma')\dot{ R}_{i3}(t, \Omega_\gamma', R_{\alpha\beta}') + (m_i-I_3\Omega_3'{R}_{i3}(t, \Omega_\gamma', R_{\alpha\beta}'))\Omega_2(t, \Omega_\gamma')].
\end{gather}
In obtaining this result were used that the quantities $\Omega'_3$, $(\Omega'_1)^2+(\Omega'_2)^2$ and $m_i$ are in the kernel of $\nabla$. Substituting (\ref{domemgan1}) and (\ref{gen3}) in the above equations, we get them through the elementary functions:

\begin{align}\label{gen1}
    \nonumber  R_{i1}(t, \boldsymbol{\Omega}', R') &= \bigg\{\bigg[\bigg(I_2\Omega_2'\frac{\sin{k't}}{\Omega_3'k'} - I_3\Omega_1'\frac{(1-\cos{k't})}{{k'^2}}\bigg)\cos{\phi't} - \bigg(I_2\Omega_1'\frac{\sin{k't}}{\Omega_3'k'} + I_3\Omega_2'\frac{(1-\cos{k't})}{{k'^2}}\bigg)\sin{\phi't}\bigg](M\boldsymbol{\Omega}', \textbf{G}_i')\\
    \nonumber &+ \bigg[\bigg(I_2\Omega_2'\cos{k't} - I_3\Omega_1'\Omega_3'\frac{\sin{k't}}{k'}\bigg)\cos{\phi't} - \bigg(I_2\Omega_1'\cos{k't} + I_3\Omega_2'\Omega_3'\frac{\sin{k't}}{k'}\bigg)\sin{\phi't}\bigg][\boldsymbol{\Omega}',\textbf{G}_i']_3 \\
     &+ (m_i - I_3\Omega_3'{R'}_{i3})(\Omega_1'\cos\phi' t + \Omega_2'\sin\phi' t)\bigg\}\bigg/[I_2({\Omega_1'}^2+ {\Omega_2'}^2)],
\end{align}

\begin{align}\label{gen2}
    \nonumber  R_{i2}(t, \boldsymbol{\Omega}', R') &= \bigg\{\bigg[\bigg(I_2\Omega_1'\frac{\sin{k't}}{\Omega_3'k'} + I_3\Omega_2'\frac{(1-\cos{k't})}{{k'^2}}\bigg)\cos{\phi't} + \bigg(I_2\Omega_2'\frac{\sin{k't}}{\Omega_3'k'} - I_3\Omega_1'\frac{(1-\cos{k't})}{{k'^2}}\bigg)\sin{\phi't}\bigg](M\boldsymbol{\Omega}', \textbf{G}_i')\\
    \nonumber &+ \bigg[\bigg(I_2\Omega_1'\cos{k't} + I_3\Omega_2'\Omega_3'\frac{\sin{k't}}{k'}\bigg)\cos{\phi't} + \bigg(I_2\Omega_2'\cos{k't} - I_3\Omega_1'\Omega_3'\frac{\sin{k't}}{k'}\bigg)\sin{\phi't}\bigg][\boldsymbol{\Omega}',\textbf{G}_i']_3 \\
    &- (m_i - I_3\Omega_3'{R'}_{i3})(-\Omega_1'\sin\phi' t +\Omega_2'\cos\phi' t)\bigg\}\bigg/[-I_2({\Omega_1'}^2+ {\Omega_2'}^2)],
\end{align}
where ${\Omega_1'}^2 + {\Omega_2'}^2\neq 0$.\\
\textbf{The case $\Omega_1' = \Omega_2' = 0$.} For the initial data with $\Omega_1'= \Omega_2'=0$, Eq. (17) implies ${\boldsymbol\Omega}(t, \Omega'_j)=(0, 0, \Omega'_3)^T$. Then the Poisson equations (11) state that the vectors ${\bf G}_i$ precess around this constant vector.
With this, the Poisson equations are rewritten as 

\begin{equation}\label{poisson3}
    \dot R_{i1} = \Omega_3'R_{i2},\ \ \ \ \ \ \dot R_{i2} = -\Omega_3'R_{i1},\ \ \ \ \ \ \dot R_{i3} =0.
\end{equation}
That is a linear differential also written in the form $\dot{\textbf{G}}_i = \Omega_3'T_3\textbf{G}_i$. That is the same as (\ref{1omega}), then the general solution to Poisson equations (\ref{poisson3}) is given by

\begin{equation}\label{aad4}
R(t, \Omega_3', R_{ij}') = 
    \begin{pmatrix}
    {R'}_{11}\cos{\Omega_3't} + {R'}_{12}\sin{\Omega_3't}\ \  & -{R'}_{11}\sin{\Omega_3't} + {R'}_{12}\cos{\Omega_3't}\ \  & {R'}_{13} \\ \\
    {R'}_{21}\cos{\Omega_3't} + {R'}_{22}\sin{\Omega_3't}\  \ & -{R'}_{21}\sin{\Omega_3't} + {R'}_{22}\cos{\Omega_3't}\ \  & {R'}_{23} \\ \\
   {R'}_{31}\cos{\Omega_3't} + {R'}_{32}\sin{\Omega_3't}\ \  & -{R'}_{31}\sin{\Omega_3't} + {R'}_{32}\cos{\Omega_3't}\ \  & {R'}_{33}
\end{pmatrix}
    .
\end{equation}

The equations (\ref{domemgan1}), (\ref{gen3}), (\ref{gen1}), (\ref{gen2}) and (\ref{aad4}) configure the general solution to the Euler-Poisson equations for the case $I_1=I_2$.\\
\textbf{General solution to Poisson equations describing a free symmetric rigid body.} As we saw previously, the motion of the rigid body corresponds to the solutions (\ref{gen3}), (\ref{gen1}) and (\ref{gen2}) at the point $(t, m_\gamma/I_\gamma, \delta_{\alpha\beta})$. Besides, the final expression for the rotation matrix acquires a more transparent form if we adjust the orientations of the Laboratory axes and the vector of conserved angular momentum. Since $\textbf{R}_i(0)$ are the eigenvectors of the inertia tensor $I$, we can arbitrarily set $I\textbf{R}_1(0) = I_2\textbf{R}_1(0)$ and $I\textbf{R}_2(0) = I_2\textbf{R}_2(0)$. Then any linear combination of the vectors $\textbf{R}_1(0)$ and $\textbf{R}_2(0)$ is an eigenvector of $I$ with eigenvalue $I_2$. Furthermore, since we have $\textbf{R}_i(0) = \textbf{e}_i$, we can freely rotate the vectors $\textbf{e}_1, \textbf{e}_2$ (generating a new orthonormal basis $\{\textbf{e}_1', \textbf{e}_2', \textbf{e}_3\}$) from the Laboratory basis until the fixed vector $\textbf{m}$ (written in the old basis $\textbf{e}_i$) belongs to the plane generated by $\textbf{e}_2', \textbf{e}_3$ without breaking the diagonal character of the inertia tensor $I$. Then, from the beginning, we can choose this configuration, maintaining $I$ as diagonal and having  $m_1 = 0$. Ultimately, the initial conditions to the Euler-Poisson equations are translated into the constants of the general solution to the EP equations  as $R_{ij}'= \delta_{ij}$, $\Omega_1'=0$, $\Omega_2' = m_2/I_2$ and $\Omega_3'= m_3/I_3$. Substituting these constants in (\ref{domemgan1}), (\ref{gen1}), (\ref{gen2}), (\ref{gen3}), we get the equations of motion of a free symmetric rigid body given by the angular velocity 

\begin{align}\label{o1}
    \Omega_1 = \frac{m_2}{I_2}\sin\phi t,\ \ \ \ \ \Omega_2  = \frac{m_2}{I_2}\cos\phi t,\ \ \ \ \ \ \Omega_3 = \frac{m_3}{I_3},\ \ \ \ \text{where}\ \ \phi = (I_2-I_3)m_3/I_2I_3,
\end{align}
and the rotational matrix $R(t)$

\begin{equation}\label{Rmov}
    \begin{pmatrix}
    \cos kt\cos\phi t - \hat m_3 \sin kt\sin\phi t & -\cos kt\sin\phi t - \hat m_3\sin kt\cos\phi t & \hat m_2\sin kt \\ \\
    \hat m_3\sin kt\cos\phi t + (\hat m_2^2 + \hat m_3^2\cos kt)\sin\phi t & -\hat m_3\sin kt\sin\phi t + (\hat m_2^2 + \hat m_3^2\cos kt)\cos\phi t & \hat m_2\hat m_3(1 - \cos kt) \\ \\
    -\hat m_2\sin kt\cos\phi t + \hat m_2\hat m_3(1 - \cos kt)\sin\phi t & \hat m_2\sin kt\sin\phi t + \hat m_2\hat m_3(1 - \cos kt)\cos\phi t & \hat m_3^2 + \hat m_2^2\cos kt
\end{pmatrix}
    ,
\end{equation}
that describes the rotation rotation of the points in the body about the center of mass after replacing this in (\ref{eul.1}). Above we have the frequencies $\phi = (I_2-I_3)m_3/I_2I_3$, $k = \sqrt{m_2^2 + m_3^2}/I_2 = |\textbf{m}|/I_2$ and, assuming $|\textbf{m}|\neq 0$, we denoted $\hat m_i = m_i/|\textbf{m}|$. The case $\Omega_1' = \Omega_2' = 0$ turns out to be 

\begin{equation}\label{Rmov2}
R(t) = 
    \begin{pmatrix}
    \cos{\frac{m_3}{I_3}t} \ \  & -\sin{\frac{m_3}{I_3}t}\ \  & 0  \\
    \sin{\frac{m_3}{I_3}t}\  \ & \cos{\frac{m_3}{I_3}t}\ \  & 0  \\
   0\ \  & 0\ \  & 1
\end{pmatrix}.
\end{equation}.\\
\textbf{Motion of the body.} With the purpose of an illustrative point of view, suppose $\textbf{R}_3(t) = (\hat m_2\sin kt,\ \  \hat m_2\hat m_3(1 - \cos kt),\ \  \hat m_3^2 + \hat m_2^2\cos kt)^T$ being the position vector of a particle. Then, the velocity of this particle is given by $\dot{\textbf{R}}_3(t) = k(\hat m_2\cos kt,\ \  \hat m_2\hat m_3\sin kt,\ \ -\hat m_2^2\sin kt)^T$. These two vectors are orthogonal to each other and, besides, realize that the velocity of the particle is orthogonal to the conserved angular momentum $\textbf{m} = (0, m_2, m_3)^T$: $(\textbf{m}, \dot{\textbf{R}}_3(t)) = k\hat{m}_2\sin{kt}(\hat{m}_3m_2 - \hat{m}_2m_3) = 0$. So, the particle moves in the plane with normal vector in the same direction as $\textbf{m}$. Moreover, it has a closed trajectory since the position vector is periodic, that is $\textbf{R}_3(t) = \textbf{R}_3(t + 2\pi/k)$. If $\theta$ is the angle between the vector $\textbf{R}_3(t)$ and the normal vector from the plane, then 

\begin{equation}
    \cos \theta = \frac{(\textbf{m}, \textbf{R}_3(t))}{|\textbf{m}|} = \hat{m}_2^2\hat{m}_3(1-\cos{kt}) + \hat{m}_3^3 + \hat{m}_2^2\hat{m}_3\cos{kt} = \hat{m}_3.
\end{equation}
This angle does not vary with time, then we can conclude that the particles have a circular trajectory in the plane. The radius will be given by $\sin\theta = \hat{m}_2$, since the position vector has unitary length. Besides the rotation has uniform angular frequency $|\dot{\textbf{R}}_3(t)|/\hat{m}_2 = k$.

Furthermore, since the inertia axes are always in the same direction as the vectors $\textbf{R}_i(t)$, the third inertia axis precesses with uniform angular velocity $k$ around the axis in the direction as the conserved angular momentum $\textbf{m}$ while the others inertia axes are rotating about the variable axis in direction as $\textbf{R}_3(t)$ for each instant of time. The difference between this most general motion and the motion described by (\ref{Rmov2}) (case $m_1 = m_2 = 0$) is that the third inertia axis is fixed in the same direction as the third axis of Laboratory frame. Then the others will rotates about the third fixed axis with uniform angular velocity $m_3/I_3$.

\section{General solution to Poisson equations by direct summation of series for all $R_{ij}$}\label{2.1}

In the previous section the rotational matrix $R(t)$ was obtained by making use of proprieties caused by the integrals of motion. This path masked the calculation so much that it was not clear to see where we were summing series. So, to show even more the capacity of the formula (\ref{oprint}), in this section we will directly sum the series (\ref{alexei}). To this aim, we will show explicitly how are given the terms $\nabla^n {R}_{ij}'$. When $n=0$, we have the first term $R_{ij}'$. When $n=1$, the therm $\nabla R_{ij}$' identifies itself with the Poisson equations:
\begin{eqnarray}\label{kkk}
\nabla R_{ij}' = -\epsilon_{jkm}\Omega_k'R_{im}' = [\textbf{G}_i', \boldsymbol{\Omega}']_j.
\end{eqnarray}
When $n=2$, we get the formula:

\begin{equation}\label{d2}
    \nabla^2{R'}_{ij} = -\boldsymbol{\Omega}'^2{R'}_{ij} + (B_i)_j\Omega_j',
\end{equation}
where $B_i$ are three $3\times 3$ diagonal matrices  given by

\begin{equation}\label{Bi}
    B_i=diag\big(( B_2\boldsymbol{\Omega'},\textbf{G}_i'), ( B_2\boldsymbol{\Omega'},\textbf{G}_i'), ( B_3\boldsymbol{\Omega'},\textbf{G}_i')\big),\ \ \ \ \text{where}\ \ \ \  B_2 = diag\bigg(1, 1, 2 - \frac{I_3}{I_2}\bigg)\ \ \text{and}\ \  B_3 = diag\bigg(\frac{I_3}{I_2}, \frac{I_3}{I_2}, 1\bigg).
\end{equation}
This equation is obtained by direct computation\footnote{In this derivation the vector product $[\textbf{G}'_i,T\boldsymbol{\Omega}']$ was computed explicitly.}:

\begin{align}
    \nonumber\nabla^2{R'}_{ij} &= \nabla(\epsilon_{j\alpha\beta}{R'}_{i\alpha}\Omega'_\beta) = \epsilon_{j\alpha\beta}[\textbf{G}_i', \boldsymbol{\Omega}']_\alpha\Omega'_\beta +\epsilon_{j\alpha\beta}{R'}_{i\alpha}(I^{-1}[I\boldsymbol{\Omega}', \boldsymbol{\Omega}'])_\beta = [[\textbf{G}_i', \boldsymbol{\Omega}'], \boldsymbol{\Omega}']_j + \phi'[\textbf{G}'_i, T\boldsymbol{\Omega}']_j \\
    \nonumber&=(\textbf{G}'_i, \boldsymbol{\Omega}')\Omega'_j - {\boldsymbol{\Omega}'}^2{R'}_{ij} + [(I_2-I_3)/I_2][diag({R'}_{i3}\Omega_3', {R'}_{i3}\Omega_3', -({R'}_{i1}\Omega_1'+ {R'}_{i2}\Omega_2'))]_j\Omega'_j \\ &= - {\boldsymbol{\Omega}'}^2{R'}_{ij} + (B_i)_j\Omega'_j.
\end{align}
With this result, we do not need to directly compute the next terms anymore. Indeed, previously, we saw that $\Omega_3'$ and ${\Omega_1'}^2 + {\Omega_2'}^2$ belongs to the kernel of $\nabla$. This implies that ${\boldsymbol{\Omega}'}^2 = {\Omega_i'}{\Omega_i'}$ also belongs to that set. Furthermore, we can define recursively the even orders of the terms $\nabla^n{R'}_{ij}$ as follows: 

\begin{align}\label{indpar}
    \nabla^0R_{ij}' = R_{ij}',\ \ \ \  \text{and}\ \ \ \  \nabla^{2(n+1)}R_{ij}' = -{\boldsymbol{\Omega}'}^2\nabla^{2n}R_{ij}' + \nabla^{2n}[(B_i)_j\Omega_j'],\ \ \ \  \text{for}\ \ n\geq 0.
\end{align}
The terms of odd orders are given by applying $\nabla$ to both sides of the above equations:

\begin{align}\label{indimpar}
    \nabla R_{ij}' = -\epsilon_{jkm}\Omega_k'R_{im}',\ \ \ \  \text{and}\ \ \ \  \nabla^{2(n+1)+1}R_{ij}' = -{\boldsymbol{\Omega}'}^2\nabla^{2n+1}R_{ij}' + \nabla^{2n+1}[(B_i)_j\Omega_j'],\ \ \ \  \text{for}\ \ n\geq 0.
\end{align}
The terms $\nabla^n[(B_i)_j\Omega_j']$ are computed according to the general Leibniz rule (\ref{derivadan}). For the convenience of the reader, below we have explicitly these expressions:

\begin{gather}
    \nabla^{2n}[(B_i)_j\Omega_j'] = \sum_{a=0}^n\binom{2n}{2a}\nabla^{2a}(B_i)_j\nabla^{2n - 2a}\Omega_j' + \sum_{b=0}^{n-1}\binom{2n}{2b+1}\nabla^{2b+1}(B_i)_j\nabla^{2n - (2b+1)}\Omega_j',\ \ \text{for all}\ \ n\geq 1,\label{pd1}\\
    \nabla^{2n+1}[(B_i)_j\Omega_j'] = \sum_{a=0}^n\binom{2n+1}{2a}\nabla^{2a}(B_i)_j\nabla^{2n+1 - 2a}\Omega_j' + \sum_{b=0}^{n}\binom{2n+1}{2b+1}\nabla^{2b+1}(B_i)_j\nabla^{2n - 2b}\Omega_j',\ \ \text{for all}\ \ n\geq 0,\label{pd2}
\end{gather}
where we have $\nabla^n\boldsymbol{\Omega}' = (\phi'T_3)^n\boldsymbol{\Omega} = {\phi'}^{n}T_3^{n}\boldsymbol{\Omega}'$, that is $\nabla^{2k}\Omega_j' = {\phi'}^{2k}[diag(\Omega_1',\Omega_2', 0)]_j$ (except by $n=0$) and $\nabla^{2k+1}\Omega_j' = \phi'^{2k+1}(T_3\boldsymbol{\Omega}')_j$; and the terms $\nabla^{n}(B_i)_j$ are given by the lemma 1.

Furthermore, the equations (\ref{indpar}), (\ref{indimpar}) completely determine the coefficients of the series (\ref{alexei}). Next, we will show how to use them to obtain the component $R_{11}$ in the matrix (\ref{Rmov}). We chose this component because its calculation has relatively the same difficulty as the others $R_{i1}$, $R_{i2}$ and the calculation to the $R_{i3}$ is more simple and does not show too much the applicability of our computations. To carry out this, firstly we will replace the initial conditions of a symmetric rigid body in the equations (\ref{indimpar})-(\ref{pd2}). To facilitate the computation, we considered them in the form:

\begin{equation}\label{ic}
    \Omega_1' = 0, \ \ \ \ \  \Omega_2' = \frac{m_2}{I_2} = \hat m_2 k, \ \ \ \ \ \Omega_3' = \frac{m_3}{I_3} = \phi + \hat m_3 k, \ \ \ \ \  {R'}_{ij} = \delta_{ij}, 
\end{equation}
where $\phi = (I_2-I_3)m_3/I_2I_3$, $k = \sqrt{{m_2}^2 + {m_3}^2}/I_2 = |\textbf{m}|/I_2$ and $\hat m_i = m_i/|\textbf{m}|$, with $|\textbf{m}|\neq 0$. When the initial conditions are substituted, the constants $\phi'$ and $k'$ become $\phi$ and $k$ from the matrix (\ref{Rmov}), respectively. Then, after being replaced the conditions (\ref{ic}) in the equations (\ref{pd1}) and (\ref{pd2}), we get

\begin{align}
    \nabla^{2n}[( B_2\boldsymbol{\Omega}', \textbf{G}_1')\Omega_1'] &= 
    \begin{cases}
        0, \ \ \text{for}\ \ n = 0,\\
        -2\hat {m_2}^2(-1)^n\sum_{b=0}^{n-1}\frac{2n!}{(2b+1)![2n-(2b+1)]!}k^{2b+2}\phi^{2n-2b}, \ \ \text{for}\ \ n\geq1,
    \end{cases}\\
    \nabla^{2n+1}[( B_2\boldsymbol{\Omega}', \textbf{G}_1')\Omega_1'] &= 0,\ \ \text{for} \ \ n\geq 0.
\end{align}
So, substituting this and the initial conditions (\ref{ic}) in the functions (\ref{indpar}), (\ref{indimpar}), we obtain the recursive functions 

\begin{gather}
    \nonumber \nabla^0R_{11}' = 1,\ \ \ \  \nabla^2R_{11}' = -(k^2 + \phi^2 + 2\hat m_3k\phi), \\ \nabla^{2(n+1)}R_{ij}' = -(k^2 + \phi^2 + 2\hat m_3k\phi)\nabla^{2n}R_{ij}' -2\hat {m_2}^2(-1)^n\sum_{b=0}^{n-1}\frac{2n!}{(2b+1)![2n-(2b+1)]!}k^{2b+2}\phi^{2n-2b},\ \  \text{for}\ \ n\geq 1,\label{indpar0}
\end{gather}
and 

\begin{align}\label{indimpar0}
    \nabla R_{ij}' = 0,\ \ \ \  \nabla^{2(n+1)+1}R_{ij}' = -(k^2 + \phi^2 + 2\hat m_3k\phi)\nabla^{2n+1}R_{ij}',\ \ \ \  \text{for}\ \ n\geq 0,
\end{align}
that determines all the coefficients of the series of $R_{11}$. Then, we will obtain that series by applying the above formulas: firstly, it is immediate that all the terms in (\ref{indimpar0}) are null. So, the coefficients of the series to $R_{11}(t)$ are given only by the function (\ref{indpar}). Then, the even terms described by (\ref{indpar0}) are given as follows: the two first terms are given $\nabla^0R_{11}' = 0$ and

\begin{align}
    \nabla^2{R'}_{11} &= (k^2 + \phi^2 + 2\hat m_3k\phi)
    = -2!\bigg(\frac{k^2}{2!} + \frac{\phi^2}{2!} + \hat m_3k\phi\bigg)=-2!\bigg[\sum_{k=0}^1\bigg(\frac{1}{2k!(2-2k)!}k^{2k}\phi^{2-2k}\bigg) + \hat m_3k\phi\bigg].
\end{align}
Then, for $n=1$, we have

\begin{align}
    \nonumber\nabla^4{R'}_{11} &= -(k^2 + \phi^2 + 2\hat m_3k\phi)\nabla^2{R'}_{11} + 4\hat {m_2}^2k^{2}\phi^{2}
    = (k^2 + \phi^2 + 2\hat m_3k\phi)^2 + 4\hat {m_2}^2k^{2}\phi^{2}\\
    \nonumber&=4!\bigg[\frac{k^4}{4!} + \frac{\phi^4}{4!} + \frac{k^2}{2!}\frac{\phi^2}{2!} + \hat m_3\bigg(\frac{k^3}{3!}\phi + k\frac{\phi^3}{3!}\bigg)\bigg]\\
    &=4!\bigg[\sum_{k=0}^2\bigg(\frac{1}{2k!(4-2k)!}k^{2k}\phi^{4-2k}\bigg) + \hat m_3\sum_{k=0}^1 \bigg(\frac{1}{(2k+1)![4-(2k+1)]!}k^{2k+1}\phi^{4-(2k+1)}\bigg)\bigg].
\end{align}
For $n=2$:

\begin{equation}
    \nabla^6{R'}_{11} = -6!\bigg[\sum_{k=0}^3\bigg(\frac{1}{2k!(6-2k)!}k^{2k}\phi^{6-2k}\bigg) + \hat m_3\sum_{k=0}^2\bigg( \frac{1}{(2k+1)![6-(2k+1)]!}k^{2k+1}\phi^{6-(2k+1)}\bigg)\bigg],
\end{equation}
and so on. So, the even terms obey the following pattern:

\begin{equation}\label{coef}
    \nabla^{2n}{R'}_{11} = (-1)^n2n!\bigg[\sum_{k=0}^n\bigg(\frac{1}{2k!(2n-2k)!}k^{2k}\phi^{2n-2k}\bigg) + \hat m_3\sum_{k=0}^{n-1} \bigg(\frac{1}{(2k+1)![2n-(2k+1)]!}k^{2k+1}\phi^{2n-(2k+1)}\bigg)\bigg].
\end{equation}
The above statement is confirmed by induction when substituting it in (\ref{indpar0}). Replacing that and (\ref{indimpar0}) in (\ref{alexei}), we obtain the series

\begin{equation}\label{r11}
    R_{11}(t) = \sum_{n=0}^\infty \sum_{k=0}^n\bigg(\frac{(-1)^nt^{2n}}{2k!(2n-2k)!}k^{2k}\phi^{2n-2k}\bigg) + \hat m_3\sum_{n=1}^\infty\sum_{k=0}^{n-1} \bigg(\frac{(-1)^nt^{2n}}{(2k+1)![2n-(2k+1)]!}k^{2k+1}\phi^{2n-(2k+1)}\bigg).
\end{equation}
This series can be represented in elementary functions as follows: consider the sequences 

\begin{gather}
    a_{m}(t) = \frac{(-1)^mk^{2m}t^{2m}}{2m!},\ \ \ \ \ \ 
    b_{m}(t) = \frac{(-1)^m\phi^{2m}t^{2m}}{2m!},\\
    c_{m}(t) = \frac{(-1)^{m}k^{2m+1}t^{2m+1}}{(2m+1)!},
    \ \ \ \ \ \  d_{m}(t) = \frac{(-1)^{m}\phi^{2m+1}t^{mk+1}}{(2m+1)!},
\end{gather}
and the equation (\ref{r11}) written as a function of them 

\begin{align}\label{r11'}
    R_{11}(t) = \sum_{n=0}^\infty\sum_{k=0}^na_{2k}b_{2n-2k} - \hat m_3\sum_{n=1}^\infty\sum_{k=0}^{n-1}c_{2k+1}d_{2n-(2k+1)}.
\end{align}
Since they represent the coefficients in the series of sines and cosines, they converge for any $t\in \mathbb{R}$. Besides, by the assertion of the Cauchy-Kovalevskaya theorem, the series (\ref{r11}) is locally convergent around $t=0$. Then, there exists a result \cite{bromwich2005introduction, abbott2001understanding} about the product of convergent power series that we are able to use. That says: if the power series $\sum_{n=0}^\infty a_nx^n$, $\sum_{n=0}^\infty b_nx^n$ converges (pointwise), then $\sum_{n=0}^\infty\sum_{i=0}^n (a_ix^i)(b_{n-i}x^{n-i}) = (\sum_{n=0}^\infty a_nx^n)(\sum_{n=0}^\infty b_nx^n)$ also converges. Furthermore, applying this result in (\ref{r11'}), we can write the function $R_{11}(t)$ as:

\begin{align}
    \nonumber R_{11}(t) &= \bigg(\sum_{n=0}^\infty a_n\bigg)\bigg(\sum_{n=0}^\infty b_n\bigg) - \hat m_3\bigg(\sum_{n=0}^\infty c_n\bigg)\bigg(\sum_{n=0}^\infty d_n\bigg)\\
    \nonumber &= \bigg(\sum_{n=0}^\infty\frac{(-1)^nk^{2n}t^{2n}}{2n!}\bigg)\bigg(\sum_{m=0}^\infty\frac{(-1)^m\phi^{2m}t^{2m}}{2m!}\bigg) - \hat m_3\bigg(\sum_{n=0}^\infty\frac{(-1)^nk^{2n+1}t^{2n+1}}{(2n+1)!}\bigg)\bigg(\sum_{m=0}^\infty\frac{(-1)^m\phi^{2m+1}t^{2m+1}}{(2m+1)!}\bigg)\\
    &=\cos kt\cos\phi t - \hat m_3 \sin kt\sin\phi t.
\end{align}
The same procedure can be done for the other components in (\ref{alexei}). Naturally, after doing this, we get the (\ref{Rmov}).\\
\textbf{Comment}: About the functions (\ref{indpar}), (\ref{indimpar}): we could also get the even and odd general terms $\nabla^{2n}{R'}_{ij}$, $\nabla^{2n+1}{R'}_{ij}$. For instance, this could be done for the even orders as follows: since $\nabla^2{R'}_{ij} = - {\boldsymbol{\Omega}'}^2{R'}_{ij} + (B_i)_j\Omega_j'$, then $\nabla^4{R'}_{ij} = - {\boldsymbol{\Omega}'}^2(- {\boldsymbol{\Omega}'}^2{R'}_{ij} + (B_i)_j\Omega_j') + \nabla^2[(B_i)_j\Omega_j']$, and so on. Performing mathematical induction, appear $\nabla^{2n}R_{ij}' = (-\boldsymbol{\Omega}')^nR_{ij}' + \sum_{k=0}^{n-1}(-\boldsymbol{\Omega}')^{n-1-k}\nabla^{2k}[(B_i)_j\Omega_j']$, where $n\geq 1$. Applying $\nabla$ in both sides in the latter equations, is obtained as the general term for the odd-order terms. The problem with this approach is the difficulty of writing the series $R_{i1}(t, \boldsymbol{\Omega}', R')$ and $R_{i2}(t, \boldsymbol{\Omega}', R')$ through elementary equations. Even if we simplify the series by substituting the initial conditions, it still not does help much. This same problem happens when trying to use (\ref{indpar}) and (\ref{indimpar}) to obtain the general solution to the Poisson equations.

\section{Conclusion}

In this article, we explore the dynamics of a free symmetric rigid body by solving the Euler-Poisson equations using an alternative method, distinct from the conventional use of Euler angles. The main goal was to obtain explicit solutions for the time evolution of the rotation matrix $R_{ij}(t)$ and the angular velocity components $\Omega_i(t)$ in terms of elementary functions. By leveraging the concept of integral lines for autonomous differential equations, we provide a general framework for solving the equations governing the body's motion. This approach allows us to calculate the rotation matrix $R_{ij}(t)$ in a closed form, bypassing the need for parametrization or solving the system of differential equations. For the case of a free symmetric body, the solution to the Euler-Poisson equations can be derived through summing power series. The representation through elementary function was possible with the use of the integrals of motion. The resulting rotation matrix is found to be consistent with previous work, confirming the validity of our approach.

To demonstrate the method, in Sect. 2, we obtain the known general solution for the angular velocity components $\Omega_i(t)$ by summing the corresponding power series, which results in elementary function expressions. With this result, in Sect. 3 we obtain the general solution to the Poisson equations (even for initial conditions that do not represent a rigid body) in elementary functions by summing the terms for $R_{ij}(t)$, leveraging the conserved quantities in the system. Then, we finish this section by obtaining the general solution to EP equations describing the most general motion of a free symmetric body (a regular precession). Finally, in Sect. 4, we made an alternative method to obtain this result by summing all components independently.

\section*{Acknowledgments}

The work has been supported by the Brazilian foundation CAPES (Coordenação de Aperfeiçoamento de Pessoal de Nível Superior - Brasil).

\section*{Appendix 1. Properties of the series $e^{th^k(z_0^j)\frac{\partial}{\partial z_0^k}}z_0^i$}

In this appendix, we exhibit and discuss some properties of the formula (\ref{oprint}). In this part of the text we were widely supported by the literature \cite{tung1985group, rotman2012introduction, bump2004lie, folland2020introduction, kepley2021constructive, spivak2006calculus, abbott2001understanding, duistermaat2012lie, bromwich2005introduction}. Firstly, let $V$ be an open subset of $\mathbb{R}^p$. Then, consider $p$ real analytic functions $h^k:V\to \mathbb{R}$. We define a linear differential operator $h^k(z_0^j)\frac{\partial}{\partial z_0^k}$ acting in sets of analytic functions $f:V\to \mathbb{R}$. Hence, the outcome $\big(h^k(z_0^j)\frac{\partial}{\partial z_0^k}\big)f(z_0^i)$ will also be a real analytic function. Since analytic functions are of class $C^\infty$, for any natural number $n$, we can recursively define the operator $\big(h^a(z_0^j)\frac {\partial}{\partial z^a_0}\big)^n$ with source and image in sets of analytic functions. Besides, performing mathematical induction, we see that it obeys the general Leibniz rule \cite{stewart2012calculus, olver1993applications, spivak2006calculus, fitzpatrick2009advanced}:

\begin{equation}\label{derivadan}
    \bigg(h^k(z_0^j)\frac {\partial}{\partial z^k_0}\bigg)^n(f(z_0^a)g(z_0^b)) = \sum_{i=0}^n\binom{n}{i} \bigg(h^{k^\prime}(z_0)\frac {\partial}{\partial z^{k^\prime}_0}\bigg)^i f(z^a_0)\ \bigg(h^{k^{\prime\prime}}(z_0)\frac {\partial}{\partial z^{k^{\prime\prime}}_0}\bigg)^{n-i}g(z^b_0), \ \ \ \ \text{for all} \ \ n \ \ \text{natural numbers},
\end{equation}
where $\binom{n}{i} = \frac{n!}{i!(n-i)!}$ is the binomial coefficient.

The next proposition clarifies how to appear the formula (\ref{oprint}):
\begin{proposition}\label{prop11}Suppose $V\subset \mathbb{R}^p$ is an open subset and let $h^i:V\to\mathbb{R}^p$ analytic functions, where $i = 1, ..., p$. Then, the initial value problem

\begin{equation}\label{sistema}
    \dot{z}^i = h^i(z^j),\ \ \ \ \ \ \ \ \ z^i(0) = z^i_0,\ \ \ \  \text{where}\ \ \ \ (z_0^1, z_0^2, ..., z_0^p)\in V,
\end{equation}
has a unique solution given by the series

\begin{equation}\label{serie}
    z^i(t, z^j_0) = e^{th^k(z^m_0)\frac {\partial}{\partial z^k_0}}z^i_0,
\end{equation}
which converges in some neighborhood of $0\in\mathbb{R}$.

\end{proposition}

\begin{proof}

By the Cauchy-Kovalevskaya theorem \cite{kepley2021constructive}, the system (\ref{sistema}) has a unique solution $z(t)$ analytic in some open interval $J\subset{R}$ containing $0$. Furthermore, there exists $r>0$ such that, for all $t\in(-r, r)$, the Taylor series

\begin{equation}\label{taylor}
    z^i(t) = \sum_{n=0}^\infty\frac{t^n}{n!}\frac{\dd^nz^i(t)}{\dd t^n}\bigg|_{t=0},
\end{equation}
converges absolutely and uniformly. Without loss of generality, we can consider that same interval for each $i= 1, 2, ..., p$. We will show that the above series can be written in the form (\ref{serie}). Consider the following statement:

\begin{equation}\label{inducao}
    \frac{\dd^n z^i(t)}{\dd t^n} = \bigg[\bigg(h^k(z^j)\frac{\partial}{\partial z^k}\bigg)^n z^i\bigg]\bigg|_{z=z(t)}, \ \ \text{for all}\ \ n\ \ \text{natural numbers}.
\end{equation}
This holds by induction: when $n=0$, this is reduced to the identity $z^i(t) = z^i|_{z=z(t)}$. So, for the base case, the statement is true. We assume, by hypothesis, that this is true for a natural $n$. Then the induction step follows

\begin{align}
   \frac{\dd^{n+1} z^i(t)}{\dd t^{n+1}} &= \frac{\dd}{\dd t}\bigg[\bigg(h^k(z^j)\frac{\partial}{\partial z^k}\bigg)^n z^i\bigg]\bigg|_{z=z(t)} = \frac{\dd}{\dd t}{h'}^i(z^j(t)) =  \bigg[h^k(z)\frac{\partial{h'}^i(z^j)}{\partial z^k}\bigg]\bigg|_{z=z(t)}= \bigg[\bigg(h^k(z^j)\frac{\partial}{\partial z^k}\bigg)^{n+1} z^i\bigg]\bigg|_{z=z(t)},
\end{align}
where ${h'}^i(z^j) =(h^k(z^j)\frac{\partial}{\partial z^k})^n z^i$. Then, since both the base case and induction step are true, by weak induction, the statement (\ref{inducao}) holds true. At $t=0$, it gets the form

\begin{equation}
    \frac{\dd^n z^i(t)}{\dd t^n}\bigg|_{t=0} = \bigg\{\bigg[\bigg(h^k(z^j)\frac{\partial}{\partial z^k}\bigg)^n z^i\bigg]\bigg|_{z=z(t)}\bigg\}\bigg|_{t=0} = \bigg[\bigg(h^k(z^j)\frac{\partial}{\partial z^k}\bigg)^n z^i\bigg]\bigg|_{z=z_0} = \bigg(h^k(z_0^j)\frac{\partial}{\partial z_0^k}\bigg)^n z_0^i.
\end{equation}
Substituting this latter equation in the Taylor series (\ref{taylor}), we obtain:

\begin{equation}
    z^i(t) = \sum_{n=0}^\infty\frac{t^n}{n!}\bigg(h^k(z_0^j)\frac{\partial}{\partial z_0^k}\bigg)^n z_0^i = e^{th^k(z_0^j)\frac {\partial}{\partial z^k_0}}z^i_0\equiv {z}^i(t, z_0^j).
\end{equation}

\end{proof}

Since in (\ref{serie}) we have a $p$-parametric family of curves solving each differential equation (\ref{sistema}), then that formula also represents the general solution to the system $\dot{z}^i = h^i(z^j)$. Even more, it is known that the existence and uniqueness of the solution for autonomous systems can be generalized to non-autonomous systems. So, if we had a non-autonomous system 

\begin{equation}
    \dot{z}^i = h^i(t,z^j),\ \ \ \ \ \ \ \ \ z^i(t_0) = z^i_0,\ \ \ \  \text{where}\ \ \ \ (z_0^1, z_0^2, ..., z_0^p)\in V,
\end{equation}
the Cauchy-Kovalevskaya theorem also would guarantee the existence and uniqueness of an analytic solution to the above non-autonomous system. In this case, its solution would be written as
\begin{equation}\label{dept}
    z^i(t, t_0, z_0^j) = e^{t\big(\frac{\partial}{\partial t_0} + h^k(t_0,z_0^j)\frac {\partial}{\partial z^k_0}\big)}z^i_0.
\end{equation}

In some fields of mathematics and physics, the function (\ref{oprint}) is quite used. It typically has two main properties. The first: $ z^i(0, z_0^j) = z_0^i$ is immediately identified in the equation (\ref{oprint}). The second, also named "group law", is translated in this work by the proposition \ref{grouplaw}, where we prove the equation (\ref{prop1}). To carry out that, we will need the following results:

\begin{lemma}\label{t2} The operator $e^{th^k(z_0)\frac {\partial}{\partial z^k_0}}$ has the property:

\begin{equation}\label{opregra}
    e^{th^k(z_0^\gamma)\frac{\partial}{\partial z^k_0}}(z_0^iz_0^j) = \big(e^{th^a(z_0^\alpha)\frac {\partial}{\partial z^a_0}}z_0^i\big)\big(e^{th^b(z_0^\beta)\frac {\partial}{\partial z^b_0}}z_0^j\big) =  z^i(t, z_0^\alpha) z^j(t, z_0^\beta),\ \ \ \ \text{where} \ \ \ \ 1\leq i,j\leq p.
\end{equation}

\end{lemma}

\begin{proof}

Consider the sequences $a_n(t) = \frac{t^n}{n!}\big(h^k(z_0^\alpha)\frac {\partial}{\partial z^k_0}\big)^n z_0^i$ and $b_n(t) = \frac{t^n}{n!}\big(h^k(z_0^\beta)\frac {\partial}{\partial z^k_0}\big)^n z_0^j$. The series $\sum_{n=0}^\infty a_n(t)$ and $\sum_{n=0}^\infty b_n(t)$ converges according to the proposition \ref{prop11}. Then, a result about the product of convergent series provides us with the relationship $\big(\sum_{n=0}^\infty a_n(t)\big)\big(\sum_{n=0}^\infty b_n(t)\big) = \sum_{n=0}^\infty \sum_{k=0}^n a_k(t)b_{n-k}(t)$, see the final of section \ref{2.1}. Writing explicitly, we have

\begin{align}
    \nonumber\big(e^{th^a(z_0^\alpha)\frac {\partial}{\partial z^a_0}}z_0^i\big)\big(e^{th^b(z_0^\beta)\frac {\partial}{\partial z^b_0}}z_0^j\big)  &= \bigg(\sum_{n=0}^\infty a_n(t)\bigg)\bigg(\sum_{n=0}^\infty b_n(t)\bigg) = \sum_{n=0}^\infty \sum_{k=0}^n a_k(t)b_{n-k}(t)\\ \nonumber&=\sum_{n=0}^\infty\sum_{k=0}^n\frac{t^n}{k!(n-k)!}\bigg(h^a(z_0^\alpha)\frac {\partial}{\partial z^a_0}\bigg)^k z_0^i\bigg(h^b(z_0^\beta)\frac {\partial}{\partial z^b_0}\bigg)^{n-k}z_0^j\\
    &=\sum_{n=0}^\infty\frac{t^n}{n!}\bigg(h^k(z_0^\gamma)\frac {\partial}{\partial z^k_0}\bigg)^{n}(z_0^iz_0^j) = e^{th^k(z_0^\gamma)\frac {\partial}{\partial z^k_0}}(z_0^iz_0^j),\label{p3}
\end{align}

where in the last equality was used the general Leibniz rule (\ref{derivadan}).

\end{proof}
The above lemma is generalized as follows:
\begin{lemma}[Generalization of the Lemma
\ref{t2}]\label{t3} Consider a product of the variables $z_0^i$:   $z_0^{i_1} ... z_0^{i_n}$, where $1\leq i_1, ..., i_n\leq p$. So, the application of the operator $e^{th^k(z_0)\frac{\partial}{\partial z_0^k}}$ in this product obeys 

\begin{equation}\label{generalizaton}
    e^{th^k(z_0^j)\frac {\partial}{\partial z^k_0}}(z_0^{i_1} ...z_0^{i_n}) = \big(e^{th^k(z_0^{j_1})\frac {\partial}{\partial z^k_0}}z_0^{i_1}\big)... \big(e^{th^k(z_0^{j_n})\frac {\partial}{\partial z^k_0}}z_0^{i_n}\big) =  z^{i_1}(t,z_0^{j_1})... z^{i_n}(t,z_0^{j_n}).
\end{equation}

\end{lemma}

\begin{proof}
    We will prove this statement by induction. The base case (for $n=1$) is true by the equation (\ref{oprint}). The case $n=2$ was proven by the previous lemma. Furthermore, suppose, by hypothesis, that it is true for some natural $n$, i.e.:

    \begin{align}\label{stment}
           z^{i_1}(t,z_0^{j_1}) z^{i_2}(t,z_0^{j_2})... z^{i_n}(t,z_0^{j_n}) = e^{th^k(z_0^{j})\frac {\partial}{\partial z^k_0}}(z_0^{i_1}z_0^{i_2}...z_0^{i_n}).
    \end{align}
    Multiplying both sides by $ z^{i_{n+1}} = e^{th^k(z_0^{j_{n+1}})\frac{\partial}{\partial z_0^k}}z_0^{i_{n+1}}$, with $1\leq i_{n+1}\leq p$, we get

    \begin{equation}
         z^{i_1}(t,z_0^{j_1}) z^{i_2}(t,z_0)... z^{i_n}(t,z_0) z^{i_{n+1}}(t,z_0) = e^{th^k(z_0^{j})\frac {\partial}{\partial z^k_0}}(z_0^{i_1}z_0^{i_2}...z_0^{i_n})e^{th^k(z_0^{j_{n+1}})\frac {\partial}{\partial z^k_0}}z_0^{i_{n+1}}.
    \end{equation}
    Next, consider the sequences $a_n(t) = \frac{t^n}{n!}\big(h^a(z_0^{j})\frac {\partial}{\partial z^a_0}\big)^n (z_0^{i_1}z_0^{i_2}...z_0^{i_n})$ and $b_n(t) = \frac{t^n}{n!}\big(h^k(z_0^{j_{n+1}})\frac {\partial}{\partial z^k_0}\big)^n z_0^{i_{n+1}}$. The series with these general terms are convergent by the hypothesis of induction\footnote{Since statement (\ref{stment}) is true by hypothesis and all $ z^{i}(t,z_0^j)$ are convergent series, we know from the theory of power series that a finite product of convergent power series is also a convergent power series. Then $e^{th^k(z_0^{j})\frac {\partial}{\partial z^k_0}}(z_0^{i_1}z_0^{i_2}...z_0^{i_n}) = \sum_{n=0}^\infty a_n(t) = \sum_{n=0}^\infty\frac{t^n}{n!}\big(h^a(z_0^{j})\frac {\partial}{\partial z^a_0}\big)^n (z_0^{i_1}z_0^{i_2}...z_0^{i_n})$ is convergent.}, then we can use the same theorem as before and get 
    
    \begin{align}
          \nonumber  z^{i_1}(t,z_0^{j_1}) z^{i_2}(t,z_0^{j_2})... z^{i_n}(t,z_0^{j_n}) z^{i_{n+1}}(t,z_0^{j_{n+1}}) &= e^{th^k(z_0^{j})\frac {\partial}{\partial z^k_0}}(z_0^{i_1}z_0^{i_2}...z_0^{i_n})e^{th^k(z_0^{j_{n+1}})\frac {\partial}{\partial z^k_0}}z_0^{i_{n+1}}\\
          \nonumber&= \bigg(\sum_{n=0}^\infty a_n(t)\bigg)\bigg(\sum_{n=0}^\infty b_n(t)\bigg) = \sum_{n=0}^\infty \sum_{k=0}^n a_k(t)b_{n-k}(t)\\ 
          \nonumber&= \sum_{n=0}^\infty\sum_{k=0}^n\frac{t^n}{k!(n-k)!}\bigg(h^a(z_0^{j})\frac{\partial}{\partial z^a_0}\bigg)^k (z_0^{i_1}z_0^{i_2}...z_0^{i_n})\bigg(h^b(z_0^{j_{n+1}})\frac {\partial}{\partial z^b_0}\bigg)^{n-k}z_0^{i_{n+1}}\\
          \nonumber&=\sum_{n=0}^\infty\frac{t^n}{n!}\bigg(h^k(z_0^j)\frac {\partial}{\partial z^k_0}\bigg)^{n}(z_0^{i_1}z_0^{i_2}...z_0^{i_n}z_0^{i_{n+1}})\\
          &=e^{th^k(z_0^j)\frac {\partial}{\partial z^k_0}}(z_0^{i_1}z_0^{i_2}...z_0^{i_n}z_0^{i_{n+1}}),
    \end{align}
    where the general Leibniz rule (\ref{derivadan}) was used. By the above relationship, we can conclude that the statement (\ref{stment}) is true for $n+1$. Then, by weak induction, it is true for every natural.
    
\end{proof}
Now, we can move on to the proof of the propriety (\ref{prop1}):

\begin{proposition}\label{grouplaw} Consider an analytic function $f:V\to \mathbb{R}$. Then we have:

\begin{equation}
    e^{th^k(z_0^j)\frac {\partial}{\partial z^k_0}}f(z_0^i) = f(e^{th^k(z_0^j)\frac {\partial}{\partial z^k_0}}z_0^i) = f( z^i(t,z_0^j)).
\end{equation}

\end{proposition}
\begin{proof}

Without loss of generality, consider $0\in V$. Since function $f$ is analytic in $V$, then there is $r>0$ such that for all $(z_0^1, z_0^2, ..., z_0^p)\in B_r(0)$ (open ball with radius $r$ centered in $0\in \mathbb{R}^p$), the Taylor series of $f$ about $0$ is given by\footnote{In this formula and the next we denote $f(z_0^i) = f(z_0)$.} 

\begin{equation}
    f(z_0) = c_0 + \sum_{i=1}^pc_i z_0^i + \sum_{i, j= 1}^pc_{i,j}z_0^i z_0^j + \sum_{i, j, k=1}^pc_{i, j, k}z_0^iz_0^jz_0^k + ...\ \ ,
\end{equation}
converges absolutely and uniformly. The numbers $c_0$, $c_i$, $c_{i,j}$ … are the coefficients of the Taylor series. Applying the operator $e^{th^k(z_0^j)\frac {\partial}{\partial z^k_0}}$ in both sides and using the previous lemma, we get

\begin{align}
     e^{th^k(z_0)\frac {\partial}{\partial z^k_0}}f(z_0) &= c_0 + \sum_{i=1}^pc_i{z}^i(t, z_0) + \sum_{i, j=1}^pc_{i,j}{z}^i(t, z_0){z}^j(t, z_0) + \sum_{i, j, k=1}^pc_{i,j, k}{z}^i(t, z_0){z}^j(t, z_0){z}^k(t, z_0) + ... = f({z}(t, z_0)).
\end{align}
The last equality in the above equation needs some explanation. The series ${z}^i(t, z_0^j)$ are analytics functions (thus continuous) defined in a neighborhood of $0\in\mathbb{R}$ represented by an interval $(-r', r')$, with $r'>0$. Then, the image of $ z(t, z_0^j)$ by $(-r', r')$: $ z((-r', r'), z_0^j))$ is a convex subset of $\mathbb{R}^p$. As $(z_0^1, z_0^2, ..., z_0^p)$ is an interior point in $B_r(0)$, then the intersection $ z((-r', r'), z_0^j))\cap B_r(0)$ is a non-empty subset of $\mathbb{R}^p$. Furthermore, the last equality is valid and it happens for any $t$ such that $z(t, z_0^j)\in  z((-r', r', z_0^j))\cap B_r(0)$.

\end{proof}

The latter result has an interesting application. By definition, the function $h^{i}(z_0^j) = \big(h^k(z_0^j)\frac{\partial}{\partial z_0^k}\big)z_0^i$ is analytic in $V$. The operator $e^{th^k(z_0^j)\frac{\partial}{\partial z_0^k}}$ acting in that function results in:
\begin{align}
    e^{th^k(z_0^j)\frac{\partial}{\partial z_0^k}}h^{i}(z_0^j) = h^{i}(z(t, z_0^j)) = \dot{z}(t, z_0^j),
\end{align}
since $z^i(t, z_0^j)$ is the general solution to the differential equation $\dot z^i = h^i(z^j)$. 

We recall that the function $F(z^1, z^2, \ldots , z^p)$ is an integral of motion of the system (\ref{sistint}), if for any solution $z^i(t)$ we have 
\begin{eqnarray}\label{aa1}
\frac{\dd}{\dd t}F(z^i(t))=0. 
\end{eqnarray}
There is a relationship among the integrals of motion and the kernel of the operator $h^k(z_0^j)\frac{\partial}{\partial z_0^k}$. \par

\begin{proposition}
    The following two conditions turn out to be equivalent: \par
    \noindent  1. $F(z^i)$ is an integral of motion. \par 

\noindent  2. $F(z^i_0)$ lies in the kernel of the operator $h^k(z_0^j)\frac{\partial}{\partial z_0^k}$, that is  $h^k(z_0^j)\frac{\partial}{\partial z_0^k} F(z^i_0)=0$.
\end{proposition}

\begin{proof}
    The condition $1$ implies the condition $2$. Indeed, being $F(z^i)$ an integral of motion, then 

    \begin{align}
        0 = \dot{F}(z^i(t, z_0^j)) = \dot{z}^k(t, z_0^j)\frac{\partial F(z^i)}{\partial z^k}\bigg|_{z = z(t, z_0)} = \bigg(h^k(z^j)\frac{\partial F(z^i)}{\partial z^k}\bigg)\bigg|_{z = z(t, z_0)},
    \end{align}
    for any $t$. So, we have $h^k(z^j)\frac{\partial F(z^i)}{\partial z^k} = 0$, concluding that $F(z^j)$ belongs to the kernel of $h^k(z^j)\frac{\partial}{\partial z^k}$.

    The condition $2$ implies the condition $1$. Indeed, consider a function $F(z_0^i)$ such that $h^k(z_0^j)\frac{\partial}{\partial z_0^k} F(z^i_0)=0$. By direct application and with the proposition \ref{grouplaw}, we have
    \begin{equation}
    F(z^i(t, z_0^j)) =  e^{th^k(z_0^j)\frac{\partial}{\partial z_0^k}}F(z_0^i) = F(z_0^i).
    \end{equation}
    The derivative relative to $t$ of the above equation provides us with $\dot{F}(z^i(t, z_0^j)) = \frac{\dd}{\dd t}F(z_0^i) = 0$. So, it is an integral of motion.
\end{proof}
In particular, if $F(z^i)$ is an integral motion, given a differentiable function $G:\mathbb{R}\to \mathbb{R}$ and a solution $z^i(t)$, we have

\begin{equation}
    \frac{\dd}{\dd t}G(F(z^i(t))) = \frac{\dd G(\alpha)}{\dd \alpha}\bigg|_{\alpha = F(z^i(t))}\frac{\dd F(z^i(t))}{\dd t} = 0.
\end{equation}
So, the above lemma results in $h^k(z_0^j)\frac{\partial}{\partial z_0^k} G(F(z^i_0))=0$.

\begin{thebibliography}{99}
    \bibitem{euler1776formulae} L. Euler, {\it Formulae generales pro translatione quacunque corporum rigidorum}, Novi Commentarii academiae scientiarum Petropolitanae {\bf 20} 189-207, (1776).

    \bibitem{arnol2013mathematical} V. I. Arnold, {\it Mathematical methods of classical mechanics} (2nd ed.), Graduate Texts in Mathematics {\bf 60}, (Springer Science \& Business Media, 1989).

    \bibitem{Deriglazov_2023} A. A. Deriglazov, {\it Lagrangian and Hamiltonian formulations of asymmetric rigid body, considered as a constrained system}, European Journal of Physics {\bf 44}, (2023) 065001, DOI=10.1088/1361-6404/ace80d.

    \bibitem{goldstein2011classical} H. Goldstein, {\it Classical mechanics} (2nd ed.), (Addison-Wesley, 1980).

    \bibitem{Landau1976Mechanics} L. D. Landau, and E. M. Lifshitz, {\it Mechanics} (3nd ed.), Course of Theoretical Physics {\bf Volume 1}, (Butterworth-Heinemann, 1976).
    
    \bibitem{palais2009disorienting} B. Palais, R. Palais, and S. Rodi, {\it A disorienting look at Euler's theorem on the axis of a rotation}, The American Mathematical Monthly {\bf 116} 10 892-909, (Taylor \& Francis, 2009).

    \bibitem{deriglazov2024asymmetricalbodyexampleanalytical} A. A. Deriglazov, {\it An asymmetrical body: example of analytical solution for the rotation matrix in elementary functions and Dzhanibekov effect}, Communications in Nonlinear Science and Numerical Simulation {\bf 118} 108257, (Elsevier, 2024).

    \bibitem{deriglazov2023problem} A. A. Deriglazov, {\it Has the problem of the motion of a heavy symmetric top been solved in quadratures?}, Foundations of Physics {\bf 54} 41, (Elsevier, 2024).
    
    \bibitem{deriglazov2023general} A. A. Deriglazov, {\it General solution to the Euler-Poisson equations of a free Lagrange top directly for the rotation matrix} (2023), arXiv:2303.02431, URL https://doi.org/10.48550/arXiv.2303.02431

    \bibitem{Deriglazov_2023EP} A. A. Deriglazov, {\it Euler–Poisson equations of a dancing spinning top, integrability and examples of analytical solutions}, Communications in Nonlinear Science and Numerical Simulation {\bf 127} 107579, (Elsevier, 2023).

    \bibitem{Deriglazov_2024} A. A. Deriglazov, {\it Rotation Matrix of a Charged Symmetrical Body: One-Parameter Family of Solutions in Elementary Functions}, Universe {\bf 10} 2218-1997, DOI=10.3390/universe10060250 (MDPI AG, 2024).

    \bibitem{deriglazov2016classical} A. A. Deriglazov, {\it Classical Mechanics: Hamiltonian and Lagrangian Formalism} (2nd ed.) (Springer, 2016).
    
    \bibitem{gantumur2011math} T. Gantumur, {\it Math 580 lecture notes 2: The cauchy-kovalevskaya theorem}, (2011).

    \bibitem{kepley2021constructive} Shane Kepley and Tianhao Zhang, {\it A constructive proof of the Cauchy-Kovalevskaya theorem for ordinary differential equations}, Journal of Fixed Point Theory and Applications {\bf 23} 7, (Springer, 2021).

    \bibitem{thelwell2012cauchy} R. J. Thelwell, P. G. Warne and D. A. Warne, {\it Cauchy-Kowalevski and polynomial ordinary differential equations}, Electronic Journal of Differential Equations {\bf 2012} 11 1-8, (2012).

    \bibitem{folland2020introduction} Gerald B. Folland, {\it Introduction to partial differential equations}, Mathematical Notes {\bf 17}, (Princeton university press, 2020).

    \bibitem{evans2022partial} L. C. Evans, {\it Partial differential equations}, Graduate Studies in Mathematics {\bf 19}, (American Mathematical Society, 2020).

    \bibitem{deriglazov2024dynamics} A. A. Deriglazov, {\it Dynamics on a submanifold: intermediate formalism versus Hamiltonian reduction of Dirac bracket, and integrability}, The European Physical Journal C {\bf 84} 311, (Springer, 2024).

    \bibitem{bromwich2005introduction} T. J. I'a. Bromwich, {\it An introduction to the theory of infinite series}, (American Mathematical Soc., 2024).

    \bibitem{abbott2001understanding} S. Abbott, Stephen et al., {\it Understanding analysis}, {\bf Volume 2}, (Springer, 2001).

    \bibitem{tung1985group} W. K. Tung, {\it Group Theory in Physics: An Introduction to Symmetry Principles, Group Representations, and Special Functions in Classical and Quantum Physics}, (1985).

    \bibitem{rotman2012introduction} J. J. Rotman, {\it An introduction to the theory of groups}, Graduate Texts in Mathematics {\bf 148}, (Springer Science \& Business Media, 2012).

    \bibitem{bump2004lie} D. Bump, et. al., {\it Lie groups}, Graduate Texts in Mathematics {\bf 225}, (Springer Science \& Business Media, 2004).
    
    \bibitem{spivak2006calculus} M. Spivak, {\it Calculus} (3rd ed.), (Cambridge University Press, 2006).

    \bibitem{duistermaat2012lie} J. J. Duistermaat and Johan A.C. Kolk, {\it Lie groups}, Universitext (Springer, 1999).

    \bibitem{stewart2012calculus} J. Stewart, {\it Calculus: early transcendentals} (8th ed.), (Cengage Learning, 2012).

    \bibitem{olver1993applications} J. P. Olver, {\it Applications of Lie groups to differential equations}, Graduate Texts in Mathematics {\bf 107}, (Springer Science \& Business Media, 1993).

    \bibitem{fitzpatrick2009advanced} b  P. Fitzpatrick, {\it Advanced calculus} (2nd ed.), (American Mathematical Soc., 2009).

    \bibitem{amer1} TS A, AH E, HF E-K. A novel approach to solving Euler’s nonlinear equations for a 3DOF dynamical motion of a rigid body under gyrostatic and constant torques. {\it Journal of Low Frequency Noise, Vibration and Active Control}. 2024;0(0). doi:10.1177/14613484241293859

    \bibitem{amer2} Amer T, El-Kafly H, Elneklawy A, Amer W. Modeling analysis on the influence of the gyrostatic moment on the motion of a charged rigid body subjected to constant axial torque. {\it Journal of Low Frequency Noise, Vibration and Active Control}. 2024;43(4):1593-1610. doi:10.1177/14613484241276381

    \bibitem{amer3} Amer, T.S., El-Kafly, H.F., Elneklawy, A.H. et al. Analyzing the dynamics of a charged rotating rigid body under constant torques. {\it Sci Rep} {\bf 14}, 9839 (2024). https://doi.org/10.1038/s41598-024-59857-z 

    \bibitem{amer4} Amer, T.S., El-Kafly, H.F., Elneklawy, A.H. et al. Analyzing the spatial motion of a rigid body subjected to constant body-fixed torques and gyrostatic moment. {\it Sci Rep} {\bf 14}, 5390 (2024). https://doi.org/10.1038/s41598-024-55964-z 

    \bibitem{amer5} Galal, A.A., Amer, T.S., Elneklawy, A.H. et al. Studying the influence of a gyrostatic moment on the motion of a charged rigid body containing a viscous incompressible liquid. {\it Eur. Phys. J. Plus} {\bf 138}, 959 (2023). https://doi.org/10.1140/epjp/s13360-023-04581-2

    \bibitem{amer6} Farag, A.M., Amer, T.S. \& Abady, I.M. Modeling and Analyzing the Dynamical Motion of a Rigid Body with a Spherical Cavity. {\it J. Vib. Eng. Technol.} {\bf 10}, 1637–1645 (2022). https://doi.org/10.1007/s42417-022-00470-7

    \bibitem{amer7} A.M. Farag, T.S. Amer, W.S. Amer, The periodic solutions of a symmetric charged gyrostat for a slightly relocated center of mass, {\it Alexandria Engineering Journal}, Volume {\bf 61}, Issue 9, 2022, Pages 7155-7170, ISSN 1110-0168, https://doi.org/10.1016/j.aej.2021.12.059.
    
\end{thebibliography}

\end{document}